\documentclass{class-for-drafts}

\usepackage{hyperref}
\usepackage{url}
\usepackage{amsmath,amssymb}
\usepackage{amsfonts}
\usepackage{algorithmicx}
\usepackage[ruled, linesnumbered, noend]{algorithm2e}

\usepackage{lipsum}

\usepackage{pifont}

\usepackage{graphicx}
\usepackage{paralist}
\usepackage{enumitem}

\usepackage{booktabs,multirow,tabularx,hhline}
\usepackage{makecell}
\usepackage{subfig}
\usepackage{tikz}
\usetikzlibrary{arrows.meta,decorations.pathreplacing,decorations.pathmorphing,plotmarks,shapes,matrix}
\tikzset{>={Latex[width=3pt,length=5pt]}}
\usepackage{epsfig}
\usepackage{epstopdf}

\usepackage{hyphenat}
\usepackage{xspace}

\newcommand{\SPRAY}[0]{\textsc{Spray}\xspace}

\newcommand{\PEERSIM}[0]{\textsc{PeerSim}\xspace}

\newcommand{\CBROADCAST}[0]{\textup{PC-broadcast}\xspace}

\usepackage{amsthm} 
\newtheorem{definition}{Definition}
\newtheorem{theorem}{Theorem}
\newtheorem{lemma}{Lemma}

\begin{document}

\title{Breaking the Scalability Barrier of Causal Broadcast\\for Large
  and Dynamic Systems}

\newcommand{\affLSNN}{LS2N, University of Nantes\\
  2 rue de la Houssini{\`e}re\\
  BP 92208, 44322 Nantes Cedex 3, France\\
  \url{first.last@univ-nantes.fr}}

\author{Brice N{\'e}delec, Pascal Molli, and Achour Most{\'e}faoui \aff \affLSNN}

\proceedings{}

\copyright{Copyright Brice N{\'e}delec, Pascal Molli, Achour Most{\'e}faoui}

\maketitle

\begin{abstract}  
  Many distributed protocols and applications rely on causal broadcast to ensure
  consistency criteria.  However, none of causality tracking state-of-the-art
  approaches scale in large and dynamic systems.  This paper presents a new
  non-blocking causal broadcast protocol suited for dynamic systems.  The
  proposed protocol outperforms state-of-the-art in size of messages, execution
  time complexity, and local space complexity. Most importantly, messages
  piggyback control information the size of which is constant.  We prove that
  for both static and dynamic systems.  Consequently, large and dynamic systems
  can finally afford causal broadcast.
\end{abstract}

\section{Introduction}

Causal broadcast~\cite{hadzilacos1994modular} is a fundamental building block of
many distributed applications~\cite{nakamoto2009bitcoin} such as distributed
social networks~\cite{borthakur2013petabyte}, distributed collaborative
software~\cite{nedelec2016crate,heinrich2012exploiting}, or distributed data
stores~\cite{demers1987epidemic,shapiro2011comprehensive,bailis2013bolton,lloyd2011cops,bravo2017saturn}.
Causal broadcast is a reliable broadcast where all connected processes
deliver each broadcast message exactly once following the happen
before relationship~\cite{lamport1978time,schwarz1994detecting}: when
Alice comments Bob's picture, everyone receives the comment after the
picture; unrelated events are delivered in any order.

Unfortunately, causal broadcast has proven expensive in dynamic environments
where any process can broadcast a message at any
time~\cite{charronbost1991concerning}. While gossiping constitutes an efficient
mean to disseminate messages to millions of
processes~\cite{demers1987epidemic,birman1999bimodal}, ensuring causal
delivery of these messages remains overcostly.  Using state-of-the-art
protocols, each message piggybacks a -- possibly compressed -- vector of
Lamport's
clocks~\cite{almeida2008interval,fidge1988timestamps,mattern1989virtual,singhal1992efficient}.
The message overhead increases monotonically, for entries cannot be reclaimed
without consensus. The message overhead increases linearly with the number of
processes $N$ that ever broadcast a message in the system. Several messages $W$
may differ their delivery, for preceding messages did not arrive
yet~\cite{mehdi2017slowdown}.  The delivery execution time takes linear time
$O(W.N)$ as well.  Causal broadcast protocols based on vectors eventually become
overcostly and inefficient.

To provide causal order, \cite{friedman2004causal} employs a different
strategy. Instead of piggybacking a vector in each message, processes forward
all messages exactly once using FIFO communication means. Gossip encompasses
forwarding so this does not constitute an overhead of the approach.  Messages
arrive ready so they are delivered immediately. This approach is both
lightweight and efficient. However, its scope is restricted to static systems.
In dynamic systems where processes can join, leave, add or remove communication
means, using this approach may lead to causal order violations.

In this paper, we break the scalability barrier of causal broadcast for large
and dynamic systems.  Our contribution is threefold:
\begin{itemize}[leftmargin=*]
\item We provide a powerful extension of~\cite{friedman2004causal} that extends
  its scope to dynamic systems. We prove that adding new communication means
  between processes constitutes the sole factor in causal order violation. Our
  approach solves this using bounded buffers and few control messages.
\item We provide the complexity analysis of our broadcast
  protocol. Table~\ref{table:comparison} compares our protocol to two
  representative solutions. Our approach handles dynamic systems while providing
  constant size overhead on messages, and constant delivery execution time.
\item We provide an experimentation highlighting the impact of our protocol on
  transmission delays before delivery. Indeed, to tolerate dynamicity our
  protocol temporarily disables new communication means.  The experiment shows
  that even under bad network conditions and high dynamicity, our protocol
  hardly degrades the mean transmission time before delivery.
\end{itemize}
Consequently, causal broadcast finally becomes an affordable and efficient
middleware for distributed protocols and applications in large and dynamic
systems.

\begin{table*}
  \begin{center}
  \caption{\label{table:comparison} Complexity of causal
    broadcast protocols. 
    $N$ is the number of processes that ever broadcast a message.
    $W$ is the number of received messages awaiting delivery.
    $P$ is the number of delivered messages that are temporarily kept before 
    being safely purged to forbid double delivery.
  }
  \newcommand{\cmark}{\ding{51}}
\newcommand{\xmark}{\ding{55}}

\begin{tabularx}{1.98\columnwidth}{@{}Xcccc@{}}
  & \makecell{dynamic systems} & \makecell{message overhead} & \makecell{local space consumption} &  \makecell{delivery execution time} \\ \cmidrule{2-5}
  vector-based~\cite{schwarz1994detecting} & \cmark & $O(N)$ & $O(N+W.N)$ & $O(W.N)$ \\
  FIFO+forward~\cite{friedman2004causal} & \xmark & $O(1)$ & $O(P)$ & $O(1)$ \\ \hline\hline
  \textbf{this paper} & \textbf{\cmark} & $\mathbf{O(1)}$ & $\mathbf{O(N)}$ & $\mathbf{O(1)}$ \\ 
\end{tabularx}

   \end{center}
\end{table*}

The rest of this paper is organized as follows: Section~\ref{sec:motivations}
shows the background and motivations of our work. Section~\ref{sec:proposal}
defines our model, describes our proposal, provides the corresponding proofs,
and details its complexity. Section~\ref{sec:experimentation} explains the
results of experimentation.  Section~\ref{sec:relatedwork} reviews the related
work. We conclude in Section~\ref{sec:conclusion}.

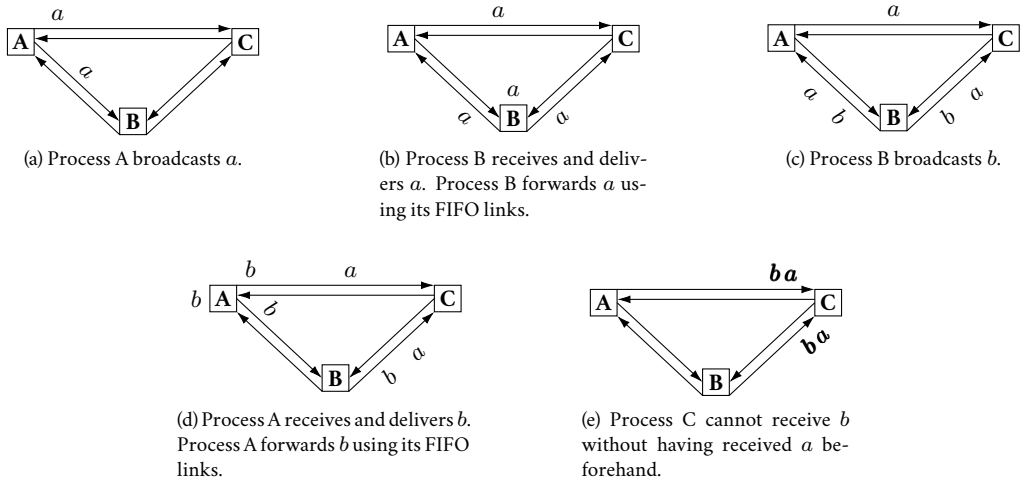
\begin{figure*}
  \begin{center}
    \subfloat[Part A][\label{fig:generalsolveA}Process~A broadcasts $a$.]
    {
\begin{tikzpicture}[scale=1]
  
  \small
  
  \newcommand\X{210/5pt};
  \newcommand\Y{30pt};

  \draw[fill=white] (0*\X, 0*\Y) node{\textbf{A}} +(-5pt, -5pt) rectangle +(5pt, 5pt);
  \draw[fill=white] (1*\X, -1*\Y) node{\textbf{B}} +(-5pt, -5pt) rectangle +(5pt, 5pt);
  \draw[fill=white] (2*\X,  0*\Y) node{\textbf{C}} +(-5pt, -5pt) rectangle +(5pt, 5pt);

  \draw[->](5+0*\X, 0*\Y) -- node[sloped, above]{$a$} (-5+1*\X, -1*\Y); 
  \draw[<-](5+0*\X, -5+0*\Y) -- (-5+1*\X, -5-1*\Y); 

  \draw[->](5+0*\X, 5+0*\Y) node[above right]{$\,\,a$} --  (-5+2*\X, 5+0*\Y); 
  \draw[<-](5+0*\X,  1.25+ 0*\Y) -- (-5+2*\X,  1.25+ 0*\Y); 
 
  \draw[<-](5+1*\X, -1*\Y) -- (-5+2*\X, 0*\Y); 
  \draw[->](5+1*\X, -5-1*\Y) -- (-5+2*\X, -5+0*\Y); 

\end{tikzpicture} }
    \hspace{40pt}
    \subfloat[Part B][\label{fig:generalsolveB}Process~B receives and 
    delivers $a$. Process~B forwards $a$ using its FIFO links.]
    {
\begin{tikzpicture}[scale=1]
  
  \small
  
  \newcommand\X{210/5pt};
  \newcommand\Y{30pt};

  \draw[fill=white] (0*\X, 0*\Y) node{\textbf{A}} +(-5pt, -5pt) rectangle +(5pt, 5pt);
  \draw[fill=white] (1*\X, -1*\Y) node{\textbf{B}} +(-5pt, -5pt) rectangle +(5pt, 5pt);
  \draw[fill=white] (2*\X,  0*\Y) node{\textbf{C}} +(-5pt, -5pt) rectangle +(5pt, 5pt);
  \draw (1*\X, 5-1*\Y) node[above]{$a$};

  \draw[->](5+0*\X, 0*\Y) -- (-5+1*\X, -1*\Y); 
  \draw[<-](5+0*\X, -5+0*\Y) -- node[sloped, below right]{$\,\,a$} (-5+1*\X, -5-1*\Y); 

  \draw[->](5+0*\X, 5+0*\Y) -- node[above left]{$a$} (-5+2*\X, 5+0*\Y); 
  \draw[<-](5+0*\X,  1.25+ 0*\Y) -- (-5+2*\X,  1.25+ 0*\Y); 
 
  \draw[<-](5+1*\X, -1*\Y) -- (-5+2*\X, 0*\Y); 
  \draw[->](5+1*\X, -5-1*\Y) --  node[sloped, below left]{$a\,\,$} (-5+2*\X, -5+0*\Y); 

\end{tikzpicture} }
    \hspace{40pt}
    \subfloat[Part C][\label{fig:generalsolveC}Process~B broadcasts $b$.]
    {
\begin{tikzpicture}[scale=1]
  
  \small
  
  \newcommand\X{210/5pt};
  \newcommand\Y{30pt};

  \draw[fill=white] (0*\X, 0*\Y) node{\textbf{A}} +(-5pt, -5pt) rectangle +(5pt, 5pt);
  \draw[fill=white] (1*\X, -1*\Y) node{\textbf{B}} +(-5pt, -5pt) rectangle +(5pt, 5pt);
  \draw[fill=white] (2*\X,  0*\Y) node{\textbf{C}} +(-5pt, -5pt) rectangle +(5pt, 5pt);

  \draw[->](5+0*\X, 0*\Y) -- (-5+1*\X, -1*\Y); 
  \draw[<-](5+0*\X, -5+0*\Y) -- node[sloped, below]{$a\,\,\,\,\,\,b$}
  (-5+1*\X, -5-1*\Y); 

  \draw[->](5+0*\X, 5+0*\Y) -- node[above]{$a$} (-5+2*\X, 5+0*\Y); 
  \draw[<-](5+0*\X,  1.25+ 0*\Y) -- (-5+2*\X,  1.25+ 0*\Y); 
 
  \draw[<-](5+1*\X, -1*\Y) -- (-5+2*\X, 0*\Y); 
  \draw[->](5+1*\X, -5-1*\Y) --  node[sloped, below]{$b\,\,\,\,\,\,a$} (-5+2*\X, -5+0*\Y); 

\end{tikzpicture} }
    \\
    \subfloat[Part D][\label{fig:generalsolveD}Process~A receives and 
    delivers $b$. Process~A forwards $b$ using its FIFO links.]
    {
\begin{tikzpicture}[scale=1]
  
  \small
  
  \newcommand\X{210/5pt};
  \newcommand\Y{30pt};

  \draw[fill=white] (0*\X, 0*\Y) node{\textbf{A}} +(-5pt, -5pt) rectangle +(5pt, 5pt);
  \draw (-5 + 0*\X, 0*\Y) node[left]{$b$};
  \draw[fill=white] (1*\X, -1*\Y) node{\textbf{B}} +(-5pt, -5pt) rectangle +(5pt, 5pt);
  \draw[fill=white] (2*\X,  0*\Y) node{\textbf{C}} +(-5pt, -5pt) rectangle +(5pt, 5pt);

  \draw[->](5+0*\X, 0*\Y) --
  node[sloped, above left]{$b\,\,\,$}
  (-5+1*\X, -1*\Y); 
  \draw[<-](5+0*\X, -5+0*\Y) -- (-5+1*\X, -5-1*\Y); 

  \draw[->](5+0*\X, 5+0*\Y) node[above right]{$b$}-- node[above right]{$a$} (-5+2*\X, 5+0*\Y); 
  \draw[<-](5+0*\X,  1.25+ 0*\Y) -- (-5+2*\X,  1.25+ 0*\Y); 
 
  \draw[<-](5+1*\X, -1*\Y) -- (-5+2*\X, 0*\Y); 
  \draw[->](5+1*\X, -5-1*\Y) --  node[sloped, below]{$b\,\,\,\,\,\,a$} (-5+2*\X, -5+0*\Y); 

\end{tikzpicture} }
    \hspace{40pt}
    \subfloat[Part E][\label{fig:generalsolveE}Process~C cannot receive $b$
    without having received $a$ beforehand.]
    {
\begin{tikzpicture}[scale=1]
  
  \small
  
  \newcommand\X{210/5pt};
  \newcommand\Y{30pt};

  \draw[fill=white] (0*\X, 0*\Y) node{\textbf{A}} +(-5pt, -5pt) rectangle +(5pt, 5pt);
  \draw[fill=white] (1*\X, -1*\Y) node{\textbf{B}} +(-5pt, -5pt) rectangle +(5pt, 5pt);
  \draw[fill=white] (2*\X,  0*\Y) node{\textbf{C}} +(-5pt, -5pt) rectangle +(5pt, 5pt);

  \draw[->](5+0*\X, 0*\Y) -- (-5+1*\X, -1*\Y); 
  \draw[<-](5+0*\X, -5+0*\Y) -- (-5+1*\X, -5-1*\Y); 

  \draw[->](5+0*\X, 5+0*\Y) -- (-5+2*\X, 5+0*\Y) node[above left]{$\pmb{b\,a}\,\,$}; 
  \draw[<-](5+0*\X,  1.25+ 0*\Y) -- (-5+2*\X,  1.25+ 0*\Y); 
 
  \draw[<-](5+1*\X, -1*\Y) -- (-5+2*\X, 0*\Y); 
  \draw[->](5+1*\X, -5-1*\Y) -- node[sloped, below right]{$\,\,\,\,\,\pmb{b\,a}$} (-5+2*\X, -5+0*\Y); 

\end{tikzpicture} }
    \caption{\label{fig:generalsolve}Causal broadcast~\cite{friedman2004causal}
      ensures causal order.}
  \end{center}
\end{figure*}

\section{Background and motivations}
\label{sec:motivations}

Causal broadcast ensures that all connected processes deliver each broadcast
message exactly once~\cite{hadzilacos1994modular} following the happen before
relationship~\cite{lamport1978time}. If the sending of a message $m$ precedes
the sending of a message $m'$ then all processes that deliver these two messages
need to deliver $m$ before $m'$. Otherwise they can deliver them in any order.

Encoding the logical time at broadcast regarding all other broadcasts and
piggyback this control information in each broadcast message allow processes to
ensure causal order on message delivery. Instead, \cite{friedman2004causal}
uses FIFO links and systematically forwards delivered messages.  Intuitively,
the dissemination pattern automatically makes sure that no paths from a process
to another process carry messages out of causal order.

Figure~\ref{fig:generalsolve} depicts this principle. The system comprises 3
processes connected to each other with FIFO links.  In
Figure~\ref{fig:generalsolveA}, Process~A broadcasts $a$. It sends $a$ to
Process~B and Process~C. In Figure~\ref{fig:generalsolveB}, Process~B receives,
delivers, and forwards $a$. In Figure~\ref{fig:generalsolveC}, it broadcasts
$b$. Consequently, all processes must deliver $a$ before delivering $b$. In
Figure~\ref{fig:generalsolveD}, Process~A receives, delivers, and forwards
$b$. Process~A fulfills the causal order constraint between $a$ and $b$. In
Figure~\ref{fig:generalsolveE}, we see that either directly via Process~B or
indirectly via Process~A, Process~C cannot receive $b$ before $a$. Thus, it
eventually receives, delivers, and forwards the messages following causal order.

\begin{figure}
  \begin{center}
    
\begin{tikzpicture}[scale=0.65]

  \small
  
  \newcommand\X{210/5pt};
  \newcommand\Y{30pt};
  
  \draw[->] ( -1*\X, 1*\Y) -- ( -5+0*\X, 0*\Y);
  \draw[->] ( -1*\X, 0*\Y) node[anchor=east, align=center]{Rest\\of the\\network} -- ( -5+0*\X, 0*\Y);
  \draw[->] ( -1*\X, -1*\Y) -- ( -5+0*\X, 0*\Y);

  \draw[<-] ( 5*\X, 1*\Y) -- ( 5+4*\X, 0*\Y);
  \draw[<-] ( 5*\X, 0*\Y) node[anchor=west, align=center]{Rest\\of the\\network}-- ( 5+4*\X, 0*\Y);
  \draw[<-] ( 5*\X, -1*\Y) -- ( 5+4*\X, 0*\Y);

  \draw[fill=white] (0*\X, 0*\Y) node{\textbf{A}} +(-5pt, -5pt) rectangle +(5pt, 5pt);

  \draw[fill=white] (1*\X, -1*\Y) node{\textbf{B}} +(-5pt, -5pt) rectangle +(5pt, 5pt);

  \draw[fill=white] (2*\X, -2*\Y) node{\textbf{E}} +(-5pt, -5pt) rectangle +(5pt, 5pt);
  \draw[fill=white] (2*\X,  0*\Y) node{\textbf{D}} +(-5pt, -5pt) rectangle +(5pt, 5pt);
  \draw[fill=white] (2*\X,  2*\Y) node{\textbf{C}} +(-5pt, -5pt) rectangle +(5pt, 5pt);

  \draw[fill=white] (4*\X, 0*\Y) node{\textbf{F}} +(-5pt, -5pt) rectangle +(5pt, 5pt);

  \draw[->](5+0*\X, 0*\Y) -- node[sloped, above]{\pmb{$a'\,a$}} (-5+2*\X,  2*\Y); 
  \draw[->](5+0*\X, 0*\Y) -- node[sloped, above]{\pmb{$a'\,a$}} (-5+1*\X, -1*\Y); 

  \draw[->](5+1*\X, -1*\Y) -- node[sloped, above]{$\pmb{a'}\,b\,\pmb{a}$} (-5+2*\X, 0*\Y); 
  \draw[->](5+1*\X, -1*\Y) -- node[sloped, above]{$\pmb{a'}\,b\,\pmb{a}$} (-5+2*\X, -2*\Y); 

  \draw[->](5+2*\X,  2*\Y) -- node[sloped, above]{$c\,\pmb{a'\,a}$} (-5+4*\X, 0*\Y); 
  \draw[->](5+2*\X,  0*\Y) -- node[sloped, above]{$\pmb{a'}\,b\,\pmb{a}\,d$} (-5+4*\X, 0*\Y); 
  \draw[->](5+2*\X, -2*\Y) -- node[sloped, above]{$\pmb{a'}\,b\,\,e'\,e\,\pmb{a}$} (-5+4*\X, 0*\Y);

\end{tikzpicture}     \caption{\label{fig:disseminationtree}The principle of
      \cite{friedman2004causal} works in large systems where processes have
      partial knowledge of the membership.}
  \end{center}
\end{figure}
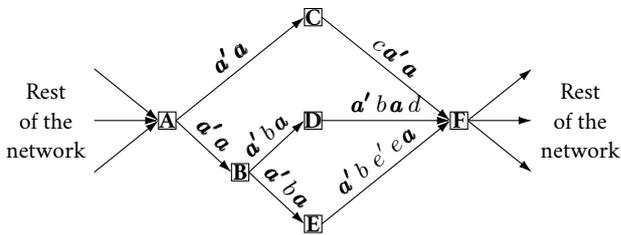

In large systems comprising from hundreds to millions of processes, no process
can afford to maintain the full membership to communicate with. Instead,
processes have a much smaller view called neighborhood. Forwarding messages
allows them to reach all members of the system, either directly or transitively
in a gossip fashion~\cite{demers1987epidemic,birman1999bimodal}. In
large systems, forwarding is mandatory.  Processes pay the price of gossiping
whatever broadcast protocol. They must create and send copies of the original
broadcast message. Since gossiping already encompasses forwarding of messages,
it does not constitute an additional overhead of~\cite{friedman2004causal}.

Figure~\ref{fig:disseminationtree} shows that such causal broadcast ensures
causal order in larger systems where processes have limited knowledge of the
membership.  Process~A only knows about Process~B and Process~C.  Yet,
Process~A's broadcast messages $a$ and $a'$ arrive to all other processes either
directly or transitively. In addition, $a$ and $a'$ always arrive in causal
order at all processes despite concurrency and whatever the dissemination path.

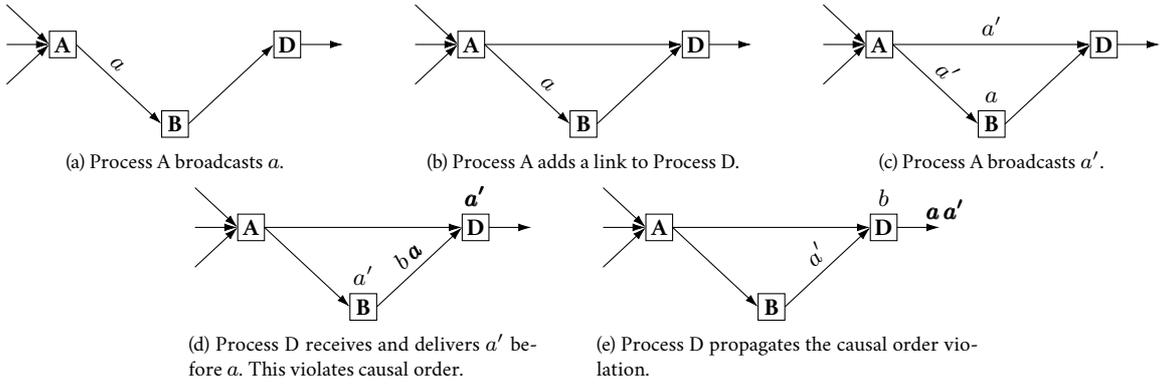
\begin{figure*}
  \begin{center}
    \subfloat[Part 1][\label{fig:preventiveproblemA}Process~A broadcasts $a$.]
    {
\begin{tikzpicture}[scale=1]
  
  \small
  
  \newcommand\X{210/5pt};
  \newcommand\Y{30pt};
  
  \draw[->] ( -0.5*\X, 0.5*\Y) -- ( -5+0*\X, 0*\Y);
  \draw[->] ( -0.5*\X, 0*\Y) -- ( -5+0*\X, 0*\Y);
  \draw[->] ( -0.5*\X, -0.5*\Y) -- ( -5+0*\X, 0*\Y);  

  \draw[fill=white] (0*\X, 0*\Y) node{\textbf{A}} +(-5pt, -5pt) rectangle +(5pt, 5pt);
  \draw[fill=white] (1*\X, -1*\Y) node{\textbf{B}} +(-5pt, -5pt) rectangle +(5pt, 5pt);
  \draw[fill=white] (2*\X,  0*\Y) node{\textbf{D}} +(-5pt, -5pt) rectangle +(5pt, 5pt);

  \draw[->](5+0*\X, 0*\Y) -- node[sloped, above left]{$a$} (-5+1*\X, -1*\Y); 
  \draw[->](5+1*\X, -1*\Y) -- (-5+2*\X, 0*\Y); 

  \draw[->](5+2*\X, 0*\Y) -- ( 2.5*\X, 0*\Y);
\end{tikzpicture} }
    \hspace{20pt}
    \subfloat[Part 2][\label{fig:preventiveproblemA2}Process~A adds a link 
    to Process~D.]
    {
\begin{tikzpicture}[scale=1]
  
  \small
  
  \newcommand\X{210/5pt};
  \newcommand\Y{30pt};
  
  \draw[->] ( -0.5*\X, 0.5*\Y) -- ( -5+0*\X, 0*\Y);
  \draw[->] ( -0.5*\X, 0*\Y) -- ( -5+0*\X, 0*\Y);
  \draw[->] ( -0.5*\X, -0.5*\Y) -- ( -5+0*\X, 0*\Y);  

  \draw[fill=white] (0*\X, 0*\Y) node{\textbf{A}} +(-5pt, -5pt) rectangle +(5pt, 5pt);
  \draw[fill=white] (1*\X, -1*\Y) node{\textbf{B}} +(-5pt, -5pt) rectangle +(5pt, 5pt);
  \draw[fill=white] (2*\X,  0*\Y) node{\textbf{D}} +(-5pt, -5pt) rectangle +(5pt, 5pt);

  \draw[->](5+0*\X, 0*\Y) -- node[sloped, above right]{$a$} (-5+1*\X, -1*\Y); 
  \draw[->](5+1*\X, -1*\Y) -- (-5+2*\X, 0*\Y); 

  \draw[->] (5+0*\X, 0*\Y) -- (-5+2*\X, 0*\Y);

  \draw[->](5+2*\X, 0*\Y) -- ( 2.5*\X, 0*\Y);
\end{tikzpicture} }
    \hspace{20pt}
    \subfloat[Part 3][\label{fig:preventiveproblemB}Process~A broadcasts $a'$.]
    {
\begin{tikzpicture}[scale=1]
  
  \small
  
  \newcommand\X{210/5pt};
  \newcommand\Y{30pt};
  
  \draw[->] ( -0.5*\X, 0.5*\Y) -- ( -5+0*\X, 0*\Y);
  \draw[->] ( -0.5*\X, 0*\Y) -- ( -5+0*\X, 0*\Y);
  \draw[->] ( -0.5*\X, -0.5*\Y) -- ( -5+0*\X, 0*\Y);  

  \draw[fill=white] (0*\X, 0*\Y) node{\textbf{A}} +(-5pt, -5pt) rectangle +(5pt, 5pt);
  \draw[fill=white] (1*\X, -1*\Y) node{\textbf{B}} +(-5pt, -5pt) rectangle +(5pt, 5pt);
  \draw (1*\X, 5-1*\Y) node[anchor=south]{$a$};
  \draw[fill=white] (2*\X,  0*\Y) node{\textbf{D}} +(-5pt, -5pt) rectangle +(5pt, 5pt);

  \draw[->](5+0*\X, 0*\Y) -- node[sloped, above]{$a'$} (-5+1*\X, -1*\Y); 
  \draw[->](5+1*\X, -1*\Y) -- (-5+2*\X, 0*\Y); 
  
  \draw[->] (5+0*\X, 0*\Y) -- node[anchor=south]{$a'$} (-5+2*\X, 0*\Y); 

  \draw[->](5+2*\X, 0*\Y) -- ( 2.5*\X, 0*\Y);
\end{tikzpicture} }
    \hspace{20pt}
    \subfloat[Part 4][\label{fig:preventiveproblemC}Process~D receives and
    delivers $a'$ before $a$. This violates causal order.]
    {
\begin{tikzpicture}[scale=1]
  
  \small
  
  \newcommand\X{210/5pt};
  \newcommand\Y{30pt};
  
  \draw[->] ( -0.5*\X, 0.5*\Y) -- ( -5+0*\X, 0*\Y);
  \draw[->] ( -0.5*\X, 0*\Y) -- ( -5+0*\X, 0*\Y);
  \draw[->] ( -0.5*\X, -0.5*\Y) -- ( -5+0*\X, 0*\Y);  

  \draw[fill=white] (0*\X, 0*\Y) node{\textbf{A}} +(-5pt, -5pt) rectangle +(5pt, 5pt);
  \draw[fill=white] (1*\X, -1*\Y) node{\textbf{B}} +(-5pt, -5pt) rectangle +(5pt, 5pt);
  \draw (1*\X, 5-1*\Y) node[anchor=south]{$a'$};
  \draw[fill=white] (2*\X,  0*\Y) node{\textbf{D}} +(-5pt, -5pt) rectangle +(5pt, 5pt);
  \draw (2*\X, 5-0*\Y) node[anchor=south]{$\pmb{a'}$};

  \draw[->](5+0*\X, 0*\Y) --  (-5+1*\X, -1*\Y); 
  \draw[->](5+1*\X, -1*\Y) -- node[sloped, above]{$b\,\pmb{a}$} (-5+2*\X, 0*\Y); 
  
  \draw[->] (5+0*\X, 0*\Y) -- (-5+2*\X, 0*\Y); 

  \draw[->](5+2*\X, 0*\Y) -- ( 2.5*\X, 0*\Y);

\end{tikzpicture} }
    \hspace{20pt}
    \subfloat[Part 4][\label{fig:preventiveproblemD}Process~D propagates
    the causal order violation.]
    {
\begin{tikzpicture}[scale=1]
  
  \small
  
  \newcommand\X{210/5pt};
  \newcommand\Y{30pt};
  
  \draw[->] ( -0.5*\X, 0.5*\Y) -- ( -5+0*\X, 0*\Y);
  \draw[->] ( -0.5*\X, 0*\Y) -- ( -5+0*\X, 0*\Y);
  \draw[->] ( -0.5*\X, -0.5*\Y) -- ( -5+0*\X, 0*\Y);  

  \draw[fill=white] (0*\X, 0*\Y) node{\textbf{A}} +(-5pt, -5pt) rectangle +(5pt, 5pt);
  \draw[fill=white] (1*\X, -1*\Y) node{\textbf{B}} +(-5pt, -5pt) rectangle +(5pt, 5pt);
  \draw[fill=white] (2*\X,  0*\Y) node{\textbf{D}} +(-5pt, -5pt) rectangle +(5pt, 5pt);
  \draw (2*\X, 5-0*\Y) node[anchor=south]{$b$};

  \draw[->](5+0*\X, 0*\Y) --  (-5+1*\X, -1*\Y); 
  \draw[->](5+1*\X, -1*\Y) -- node[sloped, above]{$a'$} (-5+2*\X, 0*\Y); 
  
  \draw[->] (5+0*\X, 0*\Y) -- (-5+2*\X, 0*\Y); 

  \draw[->](5+2*\X, 0*\Y) -- node[above right]{$\pmb{a\,a'}$} ( 2.5*\X, 0*\Y);

\end{tikzpicture} }
    \caption{\label{fig:preventiveproblem}Causal
      broadcast~\cite{friedman2004causal} may violate causal order in dynamic
      settings.}
  \end{center}
\end{figure*}

Unfortunately, \cite{friedman2004causal} ensures causal order only in static
systems where the membership does not change and no links are added or
removed. These are not practical assumptions.  In practice, processes may join
and leave the system at any time; and processes may reconfigure their
neighborhood at any time~\cite{nedelec2016crate}.
Figure~\ref{fig:preventiveproblem} shows an example of message dissemination in
dynamic settings where causal delivery is violated. In
Figure~\ref{fig:preventiveproblemA}, Process~A broadcasts $a$. It sends $a$ to
all its neighbors. Here, it sends $a$ to Process~B only.  Afterwards, in
Figure~\ref{fig:preventiveproblemA2}, Process~A adds a link to
Process~D. Message $a$ is still traveling. In particular, it did not reach
Process~D yet. In Figure~\ref{fig:preventiveproblemB}, Process~A broadcasts
$a'$. In this example, messages travel faster using the direct link from A to D
than using B as intermediate.  We see in Figure~\ref{fig:preventiveproblemC}
that $a'$ arrives at Process~D before $a$. Figure~\ref{fig:preventiveproblemD}
shows that not only it violates causal delivery but also propagates the
violation to all processes downstream.

The causal broadcast presented in this paper extends~\cite{friedman2004causal}
and solves the causal order violation issue of dynamic systems.
Table~\ref{table:comparison} shows its complexity. Most importantly, message
overhead and delivery execution time remain constant, i.e., our approach is both
lightweight in terms of generated traffic and efficient. The local space
complexity is linear in terms of number of processes that ever broadcast a
message, and awaiting messages.  The local space complexity also comprises a
data structure to ensure causal order. We show in
Algorithm~\ref{algo:boundingbuffer} that the size of this structure can be
bounded even in presence of system failures, such as crashes.  This makes causal
broadcast an affordable and efficient middleware for distributed protocols and
applications even in large and dynamic systems.

The next section describes the proposed causal broadcast. It details its
operation, provides the proofs that it works in both static and dynamic
settings, and shows its complexity analysis.

\section{Causal broadcast\\for large and dynamic systems}
\label{sec:proposal}

In this section, we introduce \CBROADCAST (stands for Preventive Causal
broadcast), a causal broadcast protocol that breaks the scalability barrier for
large and dynamic systems.  Our approach is preventive: instead of repairing
causal order violations or reordering received messages, it simply makes sure
that messages never arrive out of causal order. Processes can immediately
deliver messages upon receipt. This not only removes most of control information
piggybacked in broadcast messages, but also leads to constant delivery execution
time. Protocols and applications can finally afford causal broadcast in large
and dynamic systems without loss of efficiency.

\subsection{Model}

A distributed system comprise processes. Processes can communicate with each
other using messages. They may not have full knowledge of the membership, for
maintenance is too costly in large and dynamic systems. Instead, processes build
overlay networks with local partial view the size of which is generally much
smaller than the actual size of the
system~\cite{bertier-d2ht,jelasity2007gossip,jelasity2009tman}. Overlay networks
can be built on top of other overlay networks.  For the rest of this paper, we
will speak of distributed systems, overlay networks, or networks indifferently.

\begin{definition}[Overlay network]
  An overlay network $G$ is a pair $\langle P,\, E\rangle$ where $P$ is a set of
  processes, and $E$ is a set of links $E: P\times P$. An overlay network is
  static when $P$ and $E$ are immutable, otherwise it
  is dynamic.
\end{definition}

\begin{definition}[Process]
  A process runs a set of instructions sequentially and communicates
  with other processes using message passing. \\
  A process' neighborhood is the set of links departing from it. \\
  A process~A can send messages to another process~B in its neighborhood:
  $s_{AB}(m)$; receive a message from another process~C that has Process~A as
  neighbor:
  $r_{AC}(m)$. \\
  A process is faulty if it crashes, otherwise it is correct. We do not consider
  byzantine processes.
\end{definition}

\begin{definition}[Unpartitioned network]
  A network is unpartitioned if and only if for any pair of correct processes,
  there exist a path -- a link or a sequence of links -- of correct processes
  between them. We only consider unpartitioned overlay networks.
\end{definition}

Causal broadcast is a communication primitive that relies on reliable broadcast
to send messages to all processes in the system.

\begin{definition}[Uniform reliable broadcast]
  When a process~A broadcasts a message $b_A(m)$, each correct process~B in the
  network eventually receives it $r_B(m)$ and delivers it $d_B(m)$.
  Uniform reliable broadcast guarantees 3 properties: \\
  \textbf{Validity:} If a correct process broadcasts a message, then it
  eventually delivers it. \\
  \textbf{Uniform Agreement:} If a process -- correct or not -- delivers a
  message, then all correct processes eventually deliver it. \\
  \textbf{Uniform Integrity:} A process delivers a message at most once, and
  only if it was previously broadcast.
\end{definition}

\begin{algorithm}
  \SetKwProg{Function}{function}{}{}
\SetKwProg{Upon}{upon}{}{}
\SetKwProg{Initially}{INITIALLY:}{}{}
\SetKwProg{Dissemination}{DISSEMINATION:}{}{}

\small

\DontPrintSemicolon
\LinesNumbered

\Initially {} {
  $Q$ \tcp*{Set of processes, $p$'s neighborhood}
  $received \leftarrow \varnothing$ \tcp*{Set of received messages}
}

\BlankLine

\Dissemination{}{
  
  \Function{$\textup{R-broadcast}(m)$} { 
    $received \leftarrow received \cup m$ \;
    \lForEach {$q \in Q$} {\textup{sendTo}($q,\, m$) \tcp*[f]{broadcast}}
    \textup{R-deliver}($m$) \; 
  }

  \BlankLine
  
  \Upon{$\textup{receive}(m)$}{
    \If {$m \not \in received$} {
      $received \leftarrow received \cup m$ \;
      \lForEach {$q \in Q$} {\textup{sendTo}($q,\, m$) \tcp*[f]{forward}}
      \textup{R-deliver}($m$) \; 
    }
  }
  
}

   \caption{\label{algo:reliablebroadcast}R-broadcast at Process $p$.}
\end{algorithm}

Algorithm~\ref{algo:reliablebroadcast} shows the instructions of a uniform
reliable broadcast. It uses a structure that keeps track of received messages in
order to deliver them at most once.  Since processes may not have full
membership knowledge, processes must forward broadcast messages. Since the
network does not have partitions, processes either receive the message directly
from the broadcaster or transitively. Thus, all correct processes eventually
deliver all messages exactly once. R-broadcast ensures validity, uniform
agreement, and uniform integrity.

Causal broadcast is a reliable broadcast that also ensures a specific ordering
among message deliveries.  To define a delivery order among messages, we define
time in a logical sense using Lamport's definition~\cite{lamport1978time}.

\begin{definition}[Happen before~\cite{lamport1978time}]
  Happen before is a transitive, irreflexive, and antisymmetric relation
  $\rightarrow$ that defines a strict partial orders of events.  The sending of
  a message always precedes its receipt.
\end{definition}

To order messages broadcast by every processes, we define causal order.

\begin{definition}[Causal order]
  The delivery order of messages follows the happen before relationships of the
  corresponding broadcasts. $\forall A,\,B,\,C,\,
  b_A(m) \rightarrow b_B(m') \implies d_C(m) \rightarrow d_C(m')$
\end{definition}

\begin{definition}[Causal broadcast]
  Causal broadcast is a uniform reliable broadcast ensuring causal order.
\end{definition}

\begin{theorem}[\label{theo:flooding}Constraint flooding in deterministic
  overlay networks is causal~\cite{friedman2004causal}]
  In static networks, a broadcast protocol is causal if it uses FIFO links,
  forwards all broadcast messages exactly once, and uses all its outgoing links.
\end{theorem}

From Theorem~\ref{theo:flooding}, reliable broadcast from
Algorithm~\ref{algo:reliablebroadcast} is causal if communication links employed
to communicate with neighbors in $Q$ are FIFO. This holds only for static
networks where $Q$ is immutable. In practice, processes can join, leave, add or
remove links to neighbors from $Q$ at any time.

\begin{lemma}[R-broadcast is causal in dynamic systems subject to
  removals\label{lem:removals}]
  R-broadcast using FIFO links is a causal broadcast in dynamic systems where
  processes can leave the system or links can be removed.
\end{lemma}

\begin{proof}
  Removing a process from the network and removing all the incoming and outgoing
  links of this process is equivalent. Since we assume that removals do not
  create network partitions\footnote{It may create partitions infringing the
    uniform agreement property. Network partitioning constitutes an orthogonal
    problem that we do not address in this paper.}, all correct processes
  eventually receive all broadcast messages. In addition, removing a link or a
  process does not reorder causally related messages. Hence, each process
  receives and delivers each broadcast message in causal order as in static
  systems.
\end{proof}

Link removals and process departures do not endanger broadcast properties.
However, Figure~\ref{fig:preventiveproblem} shows that adding links may lead to
causal order violations. The next section describes \CBROADCAST, a causal
broadcast that handles all dynamicity.

\subsection{Causal order in dynamic systems}

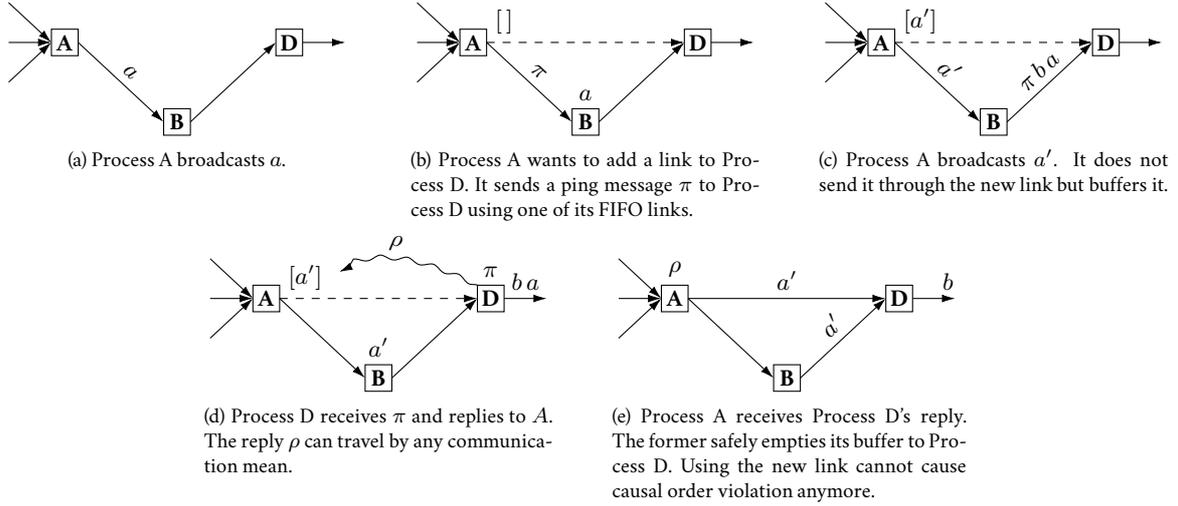
\begin{figure*}
  \begin{center}
    \subfloat[Part A][\label{fig:preventivesolveA}Process~A broadcasts $a$.]
    {
\begin{tikzpicture}[scale=1]
  
  \small
  
  \newcommand\X{210/5pt};
  \newcommand\Y{30pt};
  
  \draw[->] ( -0.5*\X, 0.5*\Y) -- ( -5+0*\X, 0*\Y);
  \draw[->] ( -0.5*\X, 0*\Y) -- ( -5+0*\X, 0*\Y);
  \draw[->] ( -0.5*\X, -0.5*\Y) -- ( -5+0*\X, 0*\Y);  

  \draw[fill=white] (0*\X, 0*\Y) node{\textbf{A}} +(-5pt, -5pt) rectangle +(5pt, 5pt);
  \draw[fill=white] (1*\X, -1*\Y) node{\textbf{B}} +(-5pt, -5pt) rectangle +(5pt, 5pt);
  \draw[fill=white] (2*\X,  0*\Y) node{\textbf{D}} +(-5pt, -5pt) rectangle +(5pt, 5pt);

  \draw[->](5+0*\X, 0*\Y) -- node[sloped, above]{$a$} (-5+1*\X, -1*\Y); 
  \draw[->](5+1*\X, -1*\Y) -- (-5+2*\X, 0*\Y); 

  \draw[->](5+2*\X, 0*\Y) -- ( 2.5*\X, 0*\Y);
\end{tikzpicture} }
    \hspace{20pt}
    \subfloat[Part B][\label{fig:preventivesolveB}Process~A wants
    to add a link to Process~D. 
    It sends a ping message $\pi$ to Process~D using one of its FIFO links.]
    {
\begin{tikzpicture}[scale=1]
  
  \small
  
  \newcommand\X{210/5pt};
  \newcommand\Y{30pt};
  
  \draw[->] ( -0.5*\X, 0.5*\Y) -- ( -5+0*\X, 0*\Y);
  \draw[->] ( -0.5*\X, 0*\Y) -- ( -5+0*\X, 0*\Y);
  \draw[->] ( -0.5*\X, -0.5*\Y) -- ( -5+0*\X, 0*\Y);  

  \draw[fill=white] (0*\X, 0*\Y) node{\textbf{A}} +(-5pt, -5pt) rectangle +(5pt, 5pt);
  \draw[fill=white] (1*\X, -1*\Y) node{\textbf{B}} +(-5pt, -5pt) rectangle +(5pt, 5pt);
  \draw (1*\X, 5-1*\Y) node[anchor=south]{$a$};
  \draw[fill=white] (2*\X,  0*\Y) node{\textbf{D}} +(-5pt, -5pt) rectangle +(5pt, 5pt);

  \draw[->](5+0*\X, 0*\Y) -- node[sloped, above]{$\pi$} (-5+1*\X, -1*\Y); 
  \draw[->](5+1*\X, -1*\Y) -- (-5+2*\X, 0*\Y); 
  
  \draw[->, dashed] (5+0*\X, 0*\Y) node[anchor=south west]{$[\,]$} -- (-5+2*\X, 0*\Y); 

  \draw[->](5+2*\X, 0*\Y) -- ( 2.5*\X, 0*\Y);
\end{tikzpicture} }
    \hspace{20pt}
    \subfloat[Part C][\label{fig:preventivesolveC}Process~A broadcasts $a'$.
    It does not send it through the new link but buffers it.]
    {
\begin{tikzpicture}[scale=1]
  
  \small
  
  \newcommand\X{210/5pt};
  \newcommand\Y{30pt};
  
  \draw[->] ( -0.5*\X, 0.5*\Y) -- ( -5+0*\X, 0*\Y);
  \draw[->] ( -0.5*\X, 0*\Y) -- ( -5+0*\X, 0*\Y);
  \draw[->] ( -0.5*\X, -0.5*\Y) -- ( -5+0*\X, 0*\Y);  

  \draw[fill=white] (0*\X, 0*\Y) node{\textbf{A}} +(-5pt, -5pt) rectangle +(5pt, 5pt);
  \draw[fill=white] (1*\X, -1*\Y) node{\textbf{B}} +(-5pt, -5pt) rectangle +(5pt, 5pt);
  \draw[fill=white] (2*\X,  0*\Y) node{\textbf{D}} +(-5pt, -5pt) rectangle +(5pt, 5pt);

  \draw[->](5+0*\X, 0*\Y) -- node[sloped, above]{$a'$} (-5+1*\X, -1*\Y); 
  \draw[->](5+1*\X, -1*\Y) -- node[sloped, above]{$\pi\,b\,a$} (-5+2*\X, 0*\Y); 
  
  \draw[->, dashed] (5+0*\X, 0*\Y) node[anchor=south west]{$[a']$} --  (-5+2*\X, 0*\Y); 

  \draw[->](5+2*\X, 0*\Y) -- ( 2.5*\X, 0*\Y);
\end{tikzpicture} }
    \hspace{20pt}
    \subfloat[Part D][\label{fig:preventivesolveD}Process~D receives
    $\pi$ and replies to $A$.
    The reply $\rho$ can travel by any communication
    mean.]
    {
\begin{tikzpicture}[scale=1]
  
  \small
  
  \newcommand\X{210/5pt};
  \newcommand\Y{30pt};
  
  \draw[->] ( -0.5*\X, 0.5*\Y) -- ( -5+0*\X, 0*\Y);
  \draw[->] ( -0.5*\X, 0*\Y) -- ( -5+0*\X, 0*\Y);
  \draw[->] ( -0.5*\X, -0.5*\Y) -- ( -5+0*\X, 0*\Y);  

  \draw[fill=white] (0*\X, 0*\Y) node{\textbf{A}} +(-5pt, -5pt) rectangle +(5pt, 5pt);
  \draw[fill=white] (1*\X, -1*\Y) node{\textbf{B}} +(-5pt, -5pt) rectangle +(5pt, 5pt);
  \draw (1*\X, 5-1*\Y) node[anchor=south]{$a'$};
  \draw[fill=white] (2*\X,  0*\Y) node{\textbf{D}} +(-5pt, -5pt) rectangle +(5pt, 5pt);
  \draw (2*\X, 5-0*\Y) node[anchor=south]{$\pi$};

  \draw[->](5+0*\X, 0*\Y) -- (-5+1*\X, -1*\Y); 
  \draw[->](5+1*\X, -1*\Y) -- (-5+2*\X, 0*\Y); 
  
  \draw[->, dashed] (5+0*\X, 0*\Y) node[anchor=south west]{$[a']$} --  (-5+2*\X, 0*\Y); 

  \draw[->, decorate, decoration={snake, amplitude=0.3mm}](-5+2*\X, 5+0*\Y)
  to[out=180-25, in=25] node[sloped, above left]{$\rho$}(0.65*\X, 10+0*\Y); 

  \draw[->](5+2*\X, 0*\Y) -- node[anchor=south]{$b\,a$}( 2.5*\X, 0*\Y);
\end{tikzpicture} }
    \hspace{20pt}
    \subfloat[Part E][\label{fig:preventivesolveE}Process~A receives
    Process~D's reply. 
    The former safely empties its buffer to Process~D. 
    Using the new link cannot cause causal order violation anymore.]
    {
\begin{tikzpicture}[scale=1]
  
  \small
  
  \newcommand\X{210/5pt};
  \newcommand\Y{30pt};
  
  \draw[->] ( -0.5*\X, 0.5*\Y) -- ( -5+0*\X, 0*\Y);
  \draw[->] ( -0.5*\X, 0*\Y) -- ( -5+0*\X, 0*\Y);
  \draw[->] ( -0.5*\X, -0.5*\Y) -- ( -5+0*\X, 0*\Y);  

  \draw[fill=white] (0*\X, 0*\Y) node{\textbf{A}} +(-5pt, -5pt) rectangle +(5pt, 5pt);
  \draw (0*\X, 5-0*\Y) node[anchor=south]{$\rho$};
  \draw[fill=white] (1*\X, -1*\Y) node{\textbf{B}} +(-5pt, -5pt) rectangle +(5pt, 5pt);
  \draw[fill=white] (2*\X,  0*\Y) node{\textbf{D}} +(-5pt, -5pt) rectangle +(5pt, 5pt);

  \draw[->](5+0*\X, 0*\Y) -- (-5+1*\X, -1*\Y); 
  \draw[->](5+1*\X, -1*\Y) -- node[sloped, above]{$a'$} (-5+2*\X, 0*\Y); 
  
  \draw[->] (5+0*\X, 0*\Y)  -- node[anchor=south]{$a'$} (-5+2*\X, 0*\Y); 

  \draw[->](5+2*\X, 0*\Y) -- node[anchor=south west]{$b$}( 2.5*\X, 0*\Y);
\end{tikzpicture} }
    \caption{\label{fig:preventivesolve}\CBROADCAST does not violate causal
      order in dynamic settings.}
  \end{center}
\end{figure*}

\CBROADCAST stands for Preventive Causal broadcast. It prevents causal order
violations by forbidding the usage of new links until proven safe. It
constitutes a powerful yet simple extension
of~\cite{friedman2004causal}. Table~\ref{table:comparison} shows that it
preserves both constant message overhead and constant delivery execution time in
dynamic settings.

Figure~\ref{fig:preventiveproblem} shows that adding links may infringe the
causal order property of causal broadcast.  New links may act as shortcut for
new messages: new messages that travel through new links may arrive before
preceding messages that took longer paths. To prevent this behavior, we define
the safety of a link. \CBROADCAST uses all and only safe links to disseminate
messages.

\begin{definition}[\label{def:safe}Safe link]
  A link from Process~A to Process~B is safe if and only if Process~B received
  or will receive all messages delivered by Process~A before receiving any
  message that Process~A will deliver:
  $safe_{AB} \equiv \forall m,\, m',\, d_A(m) \rightarrow s_{AB}(m') \implies
  r_B(m) \rightarrow r_{BA}(m')$
\end{definition}

Added links start unsafe. In Figure~\ref{fig:preventiveproblem}, Process~A uses
the link to broadcast $a'$ while it is unsafe: Process~B did not receive $a$
yet, and there was no guaranty that Process~B would receive $a$ before receiving
$a'$ from the new link. In this example, the worst happens and Process~B
receives then delivers $a'$ before $a$ which violates causal order.

The challenge is to make unsafe links safe using local knowledge only. A
straightforward mean for Process~A to achieve this consists in sending all its
delivered messages to Process~B using this unsafe link. This guarantees that any
message delivered by A will be received by B before A starts using the new --
now safe -- link for causal broadcast. However, this is costly both in local
space and generated traffic. Performing an anti-entropy round to extract missing
messages would also be overcostly in terms of generated traffic for it would
require sending the vector of received messages~\cite{demers1987epidemic}.
Instead, Process~A avoid sending most of messages by initiating a ping phase
to Process~B. 

\begin{definition}[Ping phase]
  Ping phase starts when Process~A pings Process~B. Ping messages $\pi$ travel
  using safe links. When Process~B receives this ping, it replies to
  Process~A. Replies $\rho$ travel using any communication mean. Ping phase ends
  when Process~A receives the reply of Process~B.
\end{definition}

\begin{lemma}[\label{lemma:ping}Ping phases acknowledge broadcast receipts]
  A ping phase from Process~A to Process~B acknowledges the receipt by B of all
  messages delivered by Process~A before this ping phase:
  $\forall m,\, d_A(m) \rightarrow s_A(\pi_{AB}) \wedge r_A(\rho_{AB}) \implies
  r_B(m)$
\end{lemma}

\begin{proof}
  Suppose a process~A initiates a ping phase to a process B. Suppose series of
  messages delivered by Process~A. Process~A sent these messages exactly once
  using all its outgoing safe links. Processes that will receive these message
  either already forwarded them or will forward them in their receipt
  order. Since Process~A's ping travels using safe links after these messages,
  when Process~B receives the ping, it already received all messages delivered
  by Process~A. Process~A receives Process~B's reply after Process~B received
  the ping. Consequently, when Process~A receives Process~B's reply, Process~B
  received all messages delivered by Process~A before the start of this ping
  phase.
\end{proof}

Upon receipt of Process~B's reply, Process~A has the guaranty that Process~B
received all its delivered messages preceding the ping phase. This is not
sufficient, for ping phases take time. Messages delivered during ping phase by
Process~A may not be received by Process~B yet. To fill this gap, Process~A
sends to Process~B the messages it buffered during ping phase.

\begin{definition}[Buffering]
  Process~A records in a buffer $\mathcal{B}$ all its delivered messages during
  a ping phase to Process~B.
  $\forall m,\, s_A(\pi_{AB}) \rightarrow d_A(m) \wedge d_A(m)\rightarrow
  r_A(\rho_{AB}) \Leftrightarrow m \in \mathcal{B} $
\end{definition}

\begin{lemma}[Ping phase and buffering makes safe links]
  Process~A makes an unsafe link to Process~B safe by completing a ping phase to
  Process~B then finalizing it by sending all delivered messages buffered during
  ping phase using the new link.
\end{lemma}

\begin{proof}
  Suppose series of messages $m_1 \ldots m_i \ldots m_j$ delivered by a
  process~A. Suppose Process~A initiated a ping phase to a process B after
  delivering $m_i$. Suppose Process~A receives Process~B's reply after $m_j$.
  We must show that when Process~A delivers a message after $m_k$, Process~B
  received or will receive $m_1 \ldots m_j$ before. \\
  From Lemma~\ref{lemma:ping}, when Process~A receives Process~B's reply,
  Process~B received $m_1 \ldots m_i$. \\
  Since Process~A buffered all messages delivered since the beginning of the
  ping phase, the buffer contains $m_{i+1} \ldots m_j$ when the ping phase ends.
  Since links are FIFO, sending messages of this buffer using the new link
  guarantees that Process~B will receive them before receiving any $m_k$
  delivered after $m_j$. The link from Process~A to Process~B became safe.
\end{proof}

\begin{lemma}[\CBROADCAST is causal in dynamic systems subject to
  additions\label{lem:additions}]
  In dynamic systems where processes can join or add links, broadcasting using
  all and only safe FIFO links ensures causal order. Without partition, the
  broadcast is causal.
\end{lemma}

\begin{proof}
  \CBROADCAST ensures \textbf{validity}, \textbf{uniform agreement}, and
  \textbf{uniform integrity},
  for it extends R-broadcast that ensures all 3 properties. \\
  We must show that \CBROADCAST ensures \textbf{causal order}:
  $\forall A,\,B,\,C,\, b_A(m) \rightarrow b_B(m')
  \implies d_C(m) \rightarrow d_C(m')$. \\
  $\forall A,\,B,\, b_A(m) \rightarrow b_B(m') \Leftrightarrow d_B(m)
  \rightarrow d_B(m') \Leftrightarrow d_B(m) \rightarrow s_B(m') \Leftrightarrow
  d_B(m) \rightarrow (\forall C\in Q, s_{BC}(m'))$.
  Since all links in $Q$ are safe links, $r_c(m) \rightarrow r_{CB}(m')$ (see
  Definition~\ref{def:safe}).  Since delivery order follows first receipt order,
  $d_C(m) \rightarrow d_C(m')$. This order on message delivery transitively
  reach all correct processes as long as the network remains unpartitioned.
\end{proof}

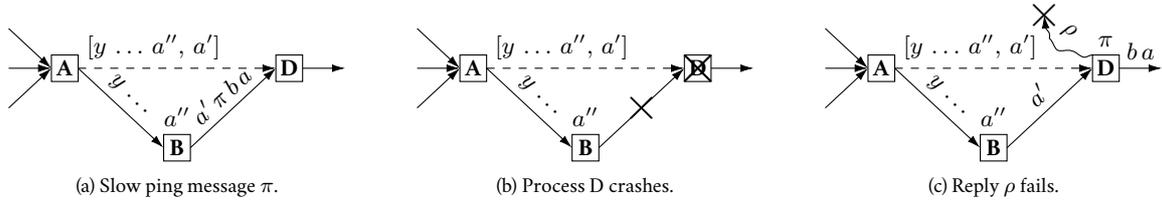
\begin{figure*}
  \begin{center}
    \subfloat[part A][\label{fig:bufferproblemA}Slow ping message $\pi$.]
    {
\begin{tikzpicture}[scale=1]
  
  \small
  
  \newcommand\X{210/5pt};
  \newcommand\Y{30pt};
  
  \draw[->] ( -0.5*\X, 0.5*\Y) -- ( -5+0*\X, 0*\Y);
  \draw[->] ( -0.5*\X, 0*\Y) -- ( -5+0*\X, 0*\Y);
  \draw[->] ( -0.5*\X, -0.5*\Y) -- ( -5+0*\X, 0*\Y);  

  \draw[fill=white] (0*\X, 0*\Y) node{\textbf{A}} +(-5pt, -5pt) rectangle +(5pt, 5pt);
  \draw[fill=white] (1*\X, -1*\Y) node{\textbf{B}} +(-5pt, -5pt) rectangle +(5pt, 5pt);
  \draw (1*\X, 5-1*\Y) node[anchor=south]{$a''$};
  \draw[fill=white] (2*\X,  0*\Y) node{\textbf{D}} +(-5pt, -5pt) rectangle +(5pt, 5pt);

  \draw[->, dashed] (5+0*\X, 0*\Y) -- (-5+2*\X, 0*\Y);

  \draw[->](5+0*\X, 0*\Y) node[anchor=south west]{$[y\, \ldots \,a'',\,a']$} -- node[sloped, above]{$y\, \dots$} (-5+1*\X, -1*\Y); 
  \draw[->](5+1*\X, -1*\Y) -- node[sloped, above]{$a'\,\pi\,b\,a$} (-5+2*\X, 0*\Y); 

  \draw[->](5+2*\X, 0*\Y) -- ( 2.5*\X, 0*\Y);
\end{tikzpicture} }
    \hspace{20pt}
    \subfloat[part B][\label{fig:bufferproblemB}Process~D crashes.]
    {
\begin{tikzpicture}[scale=1]
  
  \small
  
  \newcommand\X{210/5pt};
  \newcommand\Y{30pt};
  
  \draw[->] ( -0.5*\X, 0.5*\Y) -- ( -5+0*\X, 0*\Y);
  \draw[->] ( -0.5*\X, 0*\Y) -- ( -5+0*\X, 0*\Y);
  \draw[->] ( -0.5*\X, -0.5*\Y) -- ( -5+0*\X, 0*\Y);  

  \draw[fill=white] (0*\X, 0*\Y) node{\textbf{A}} +(-5pt, -5pt) rectangle +(5pt, 5pt);
  \draw[fill=white] (1*\X, -1*\Y) node{\textbf{B}} +(-5pt, -5pt) rectangle +(5pt, 5pt);
  \draw (1*\X, 5-1*\Y) node[anchor=south]{$a''$};
  \draw[fill=white] (2*\X,  0*\Y) node{\textbf{D}} +(-5pt, -5pt) rectangle +(5pt, 5pt);

  \draw[->, dashed] (5+0*\X, 0*\Y) -- (-5+2*\X, 0*\Y);

  \draw[->](5+0*\X, 0*\Y) node[anchor=south west]{$[y\, \ldots \,a'',\,a']$} -- node[sloped, above]{$y\,\ldots$} (-5+1*\X, -1*\Y); 
  \draw[->](5+1*\X, -1*\Y) -- 
  node{\LARGE$\times$}(-5+2*\X, 0*\Y); 

  \draw[thick] (-5+2*\X, -5+0*\Y) -- (5+2*\X, 5+0*\Y);
  \draw[thick] (-5+2*\X, 5+0*\Y) -- (5+2*\X, -5+0*\Y);

  \draw[->](5+2*\X, 0*\Y) -- ( 2.5*\X, 0*\Y);
\end{tikzpicture} }
    \hspace{20pt}
    \subfloat[part C][\label{fig:bufferproblemC}Reply $\rho$ fails.]
    {
\begin{tikzpicture}[scale=1]
  
  \small
  
  \newcommand\X{210/5pt};
  \newcommand\Y{30pt};
  
  \draw[->] ( -0.5*\X, 0.5*\Y) -- ( -5+0*\X, 0*\Y);
  \draw[->] ( -0.5*\X, 0*\Y) -- ( -5+0*\X, 0*\Y);
  \draw[->] ( -0.5*\X, -0.5*\Y) -- ( -5+0*\X, 0*\Y);  

  \draw[fill=white] (0*\X, 0*\Y) node{\textbf{A}} +(-5pt, -5pt) rectangle +(5pt, 5pt);
  \draw[fill=white] (1*\X, -1*\Y) node{\textbf{B}} +(-5pt, -5pt) rectangle +(5pt, 5pt);
  \draw (1*\X, 5-1*\Y) node[anchor=south]{$a''$};
  \draw[fill=white] (2*\X,  0*\Y) node{\textbf{D}} +(-5pt, -5pt) rectangle +(5pt, 5pt);
  \draw (2*\X, 5-0*\Y) node[anchor=south]{$\pi$};

  \draw[->, dashed] (5+0*\X, 0*\Y) -- (-5+2*\X, 0*\Y);

  \draw[->](5+0*\X, 0*\Y) node[anchor=south west]{$[y\, \ldots \,a'',\,a']$} -- node[sloped, above]{$y\,\ldots$} (-5+1*\X, -1*\Y); 
  \draw[->](5+1*\X, -1*\Y) -- node[sloped, above]{$a'$} (-5+2*\X, 0*\Y); 

  \draw[->, decorate, decoration={snake, amplitude=0.3mm}](-5+2*\X, 5+0*\Y)
  to[out=180, in=-90] node[sloped, above]{$\rho$}(1.45*\X, 20+0*\Y)
  node{\LARGE$\times$};

  \draw[->](5+2*\X, 0*\Y) -- node[above]{$b\,a$}( 2.5*\X, 0*\Y);
\end{tikzpicture} }
    \caption{\label{fig:bufferproblem}Buffers may grow unbounded due to network
      conditions.}
  \end{center}
\end{figure*}

\begin{algorithm}
  
\SetKwProg{Function}{function}{}{}
\SetKwProg{Upon}{upon}{}{}
\SetKwProg{Initially}{INITIALLY:}{}{}
\SetKwProg{Dissemination}{DISSEMINATION:}{}{}
\SetKwProg{Safety}{SAFETY:}{}{}

\small

\DontPrintSemicolon
\LinesNumbered

\Initially {} {
  $Q$ \tcp*{$p$'s neighborhood, FIFO links}
  $B \leftarrow \varnothing$ \tcp*{Map link $\rightarrow$ buffered messages}
  \BlankLine
  $counter \leftarrow 0$ \tcp*{Control message identifier}
}

\BlankLine

\Safety {} {
  \Upon{$\textup{open}(q)$} {
    \If{$|Q|>1$} {
      $counter \leftarrow counter+1$ \;
      $Q \leftarrow Q \setminus q$ \tcp*{is unsafe}
      $B[q] \leftarrow \varnothing$ \tcp*{initialize buffer}
      $\textup{ping}(p,\, q,\, counter)$ \label{line:sendlocked} \tcp*{send $\pi$} 
    }
  }
  
  \BlankLine
  
  \Upon{$\textup{receivePing}(from,\, to,\, id)$ \tcp*[f]{$to=p$}} {
    $\textup{pong}(from,\, to,\, id)$ \label{line:sendack} \tcp*{send $\rho$}
  }
  
  \BlankLine
  
  \Upon{$\textup{receivePong}(from,\, to,\, id)$ \tcp*[f]{$from=p$}} {
    \If {$to \in B$} {
      \lForEach {$m \in B[to]$} {$\textup{sendTo}(to,\, m)$
        \label{line:emptybuffer}}
      $B \leftarrow B \setminus to$ \tcp*{remove buffer}
      $Q \leftarrow Q \cup to$ \tcp*{now safe}
    }
  }
  
  \BlankLine
  
  \Upon{$\textup{close}(q)$} {
    $B \leftarrow B \setminus q$ \;
  }
}

\BlankLine

\Dissemination {} {

  \Function{$\CBROADCAST(m)$ } {
    $\textup{R-broadcast}(m)$ \label{line:rbroadcast}\;
  }
  
  \BlankLine
  
  \Upon{$\textup{R-deliver}(m)$ \label{line:rdeliver}} {
    \lForEach {$q \in B$} {$B[q] \leftarrow B[q] \cup m$
      \label{line:bufferforward} \tcp*[f]{buffers}}
    $\textup{PC-deliver}(m)$  \; 
  }

}

   \caption{\label{algo:bufferbroadcast}\CBROADCAST at Process $p$.}
\end{algorithm}

Algorithm~\ref{algo:bufferbroadcast} shows the small set of instructions that
implement safe links.  Figure~\ref{fig:preventivesolve} shows on an example how
it solves causal order violations. In Figure~\ref{fig:preventivesolveA},
Process~A broadcasts $a$.  In Figure~\ref{fig:preventivesolveB}, Process~A wants
to add a link to Process~D. It sends a ping message $\pi$ to Process~D (see
Line~\ref{line:sendlocked}) and awaits for the latter's reply.  We leave aside
the implementation of this send function (e.g. broadcast or routing).  While
awaiting, Process~A keeps its normal functioning and maintain a buffer of
messages associated with the unsafe link (see Line~\ref{line:bufferforward}). In
Figure~\ref{fig:preventivesolveC}, Process~A broadcasts another message $a'$. It
sends it normally to Process~B but does not send it to Process~D
directly. Instead, it buffers it. In Figure~\ref{fig:preventivesolveD},
Process~D receives Process~A's ping message $\pi$. Since links are FIFO, it
implicitly means that Process~D also received $a$. Process~D sends a reply
$\rho$ to Process~A (see Line~\ref{line:sendack}). $\rho$ can travel through any
communication mean. In Figure~\ref{fig:preventivesolveE}, Process~A receives
$\rho$. Consequently, Process~A knows that Process~D received and delivered at
least $a$ and all preceding messages. It empties the buffer of messages to
Process~D (see Line~\ref{line:emptybuffer}). Afterwards, the new link is
safe. Process~A uses the new link normally.

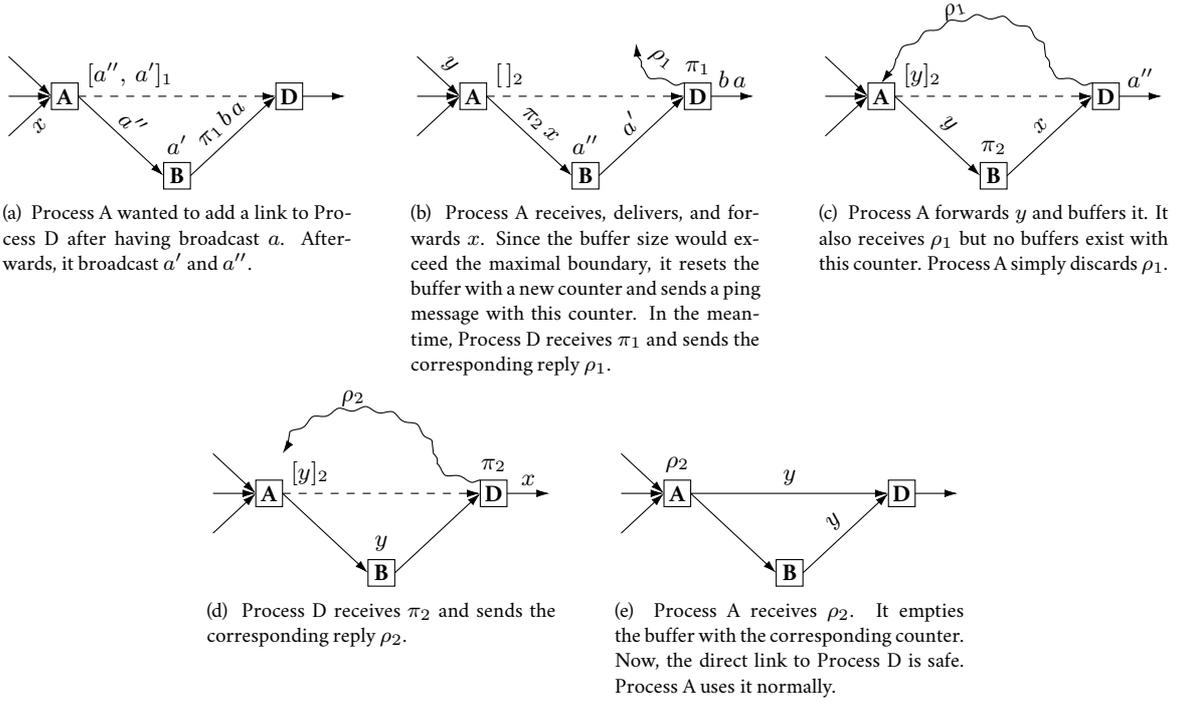
\begin{figure*}
  \begin{center}
    \subfloat[Part A][\label{fig:buffersolveA}
    Process~A wanted to add a link to Process~D after having
    broadcast $a$. Afterwards, it broadcast $a'$ and $a''$.]
    {
\begin{tikzpicture}[scale=1]
  
  \small
  
  \newcommand\X{210/5pt};
  \newcommand\Y{30pt};
  
  \draw[->] ( -0.5*\X, 0.5*\Y) -- ( -5+0*\X, 0*\Y);
  \draw[->] ( -0.5*\X, 0*\Y) -- ( -5+0*\X, 0*\Y);
  \draw[->] ( -0.5*\X, -0.5*\Y) -- node[below,sloped]{$x$} ( -5+0*\X, 0*\Y);  

  \draw[fill=white] (0*\X, 0*\Y) node{\textbf{A}} +(-5pt, -5pt) rectangle +(5pt, 5pt);
  \draw[fill=white] (1*\X, -1*\Y) node{\textbf{B}} +(-5pt, -5pt) rectangle +(5pt, 5pt);
  \draw (1*\X, 5-1*\Y) node[anchor=south]{$a'$};
  \draw[fill=white] (2*\X,  0*\Y) node{\textbf{D}} +(-5pt, -5pt) rectangle +(5pt, 5pt);

  \draw[->, dashed] (5+0*\X, 0*\Y) -- (-5+2*\X, 0*\Y);

  \draw[->](5+0*\X, 0*\Y) node[anchor=south west]{$[a'',\,a']_1$} -- 
  node[sloped, above]{$a''$} (-5+1*\X, -1*\Y); 
  \draw[->](5+1*\X, -1*\Y) -- node[sloped, above]{$\pi_1\,b\,a$} (-5+2*\X, 0*\Y); 

  \draw[->](5+2*\X, 0*\Y) -- ( 2.5*\X, 0*\Y);
\end{tikzpicture} }
    \hspace{20pt}
    \subfloat[Part B][\label{fig:buffersolveB}
    Process~A receives, delivers, and forwards $x$. Since the buffer size would 
    exceed the maximal boundary, it resets the buffer with a new counter and
    sends a ping message with this counter. In the meantime, Process~D
    receives $\pi_1$ and sends the corresponding reply $\rho_1$.]
    {
\begin{tikzpicture}[scale=1]
  
  \small
  
  \newcommand\X{210/5pt};
  \newcommand\Y{30pt};
  
  \draw[->] ( -0.5*\X, 0.5*\Y) -- node[sloped, above]{$y$}( -5+0*\X, 0*\Y);
  \draw[->] ( -0.5*\X, 0*\Y) -- ( -5+0*\X, 0*\Y);
  \draw[->] ( -0.5*\X, -0.5*\Y) -- ( -5+0*\X, 0*\Y);  

  \draw[fill=white] (0*\X, 0*\Y) node{\textbf{A}} +(-5pt, -5pt) rectangle +(5pt, 5pt);
  \draw[fill=white] (1*\X, -1*\Y) node{\textbf{B}} +(-5pt, -5pt) rectangle +(5pt, 5pt);
  \draw (1*\X, 5-1*\Y) node[anchor=south]{$a''$};
  \draw[fill=white] (2*\X,  0*\Y) node{\textbf{D}} +(-5pt, -5pt) rectangle +(5pt, 5pt);
  \draw (2*\X, 5-0*\Y) node[anchor=south]{$\pi_1$};

  \draw[->, dashed] (5+0*\X, 0*\Y) -- (-5+2*\X, 0*\Y);

  \draw[->](5+0*\X, 0*\Y) node[anchor=south west]{$[\,]_2$} -- 
  node[sloped, above]{$\pi_2\,x$} (-5+1*\X, -1*\Y); 
  \draw[->](5+1*\X, -1*\Y) -- node[sloped, above]{$a'$} (-5+2*\X, 0*\Y); 

  \draw[->, decorate, decoration={snake, amplitude=0.3mm}](-5+2*\X, 5+0*\Y)
  to[out=180, in=-90] node[sloped, above]{$\rho_1$}(1.45*\X, 20+0*\Y);

  \draw[->](5+2*\X, 0*\Y) -- node[above]{$b\,a$}( 2.5*\X, 0*\Y);
\end{tikzpicture} }
    \hspace{20pt}
    \subfloat[Part C][\label{fig:buffersolveC}
    Process~A forwards $y$ and buffers it. It also
    receives $\rho_1$ but no buffers exist with this counter. Process~A 
    simply discards $\rho_1$.]
    {
\begin{tikzpicture}[scale=1]
  
  \small
  
  \newcommand\X{210/5pt};
  \newcommand\Y{30pt};
  
  \draw[->] ( -0.5*\X, 0.5*\Y) -- ( -5+0*\X, 0*\Y);
  \draw[->] ( -0.5*\X, 0*\Y) -- ( -5+0*\X, 0*\Y);
  \draw[->] ( -0.5*\X, -0.5*\Y) -- ( -5+0*\X, 0*\Y);  

  \draw[fill=white] (0*\X, 0*\Y) node{\textbf{A}} +(-5pt, -5pt) rectangle +(5pt, 5pt);
  \draw[fill=white] (1*\X, -1*\Y) node{\textbf{B}} +(-5pt, -5pt) rectangle +(5pt, 5pt);
  \draw (1*\X, 5-1*\Y) node[anchor=south]{$\pi_2$};
  \draw[fill=white] (2*\X,  0*\Y) node{\textbf{D}} +(-5pt, -5pt) rectangle +(5pt, 5pt);

  \draw[->, dashed] (5+0*\X, 0*\Y) -- (-5+2*\X, 0*\Y);

  \draw[->](5+0*\X, 0*\Y) node[anchor=south west]{$[y]_2$} -- 
  node[sloped, above]{$y$}
  (-5+1*\X, -1*\Y); 
  \draw[->](5+1*\X, -1*\Y) -- node[sloped, above]{$x$} (-5+2*\X, 0*\Y); 

  \draw[decorate, decoration={snake, amplitude=0.3mm}](-5+2*\X, 5+0*\Y)
  to[out=180, in=-90] (1.45*\X, 20+0*\Y);
  \draw[->, decorate, decoration={snake, amplitude=0.3mm}](1.45*\X, 20+0*\Y)
  to[out=90+45, in=90-35] node[sloped, above]{$\rho_1$}(0*\X, 5+0*\Y);

  \draw[->](5+2*\X, 0*\Y) -- node[above]{$a''$}( 2.5*\X, 0*\Y);
\end{tikzpicture} }
    \hspace{20pt}
    \subfloat[Part D][\label{fig:buffersolveD}
    Process~D receives $\pi_2$ and sends the corresponding
    reply $\rho_2$.]
    {
\begin{tikzpicture}[scale=1]
  
  \small
  
  \newcommand\X{210/5pt};
  \newcommand\Y{30pt};
  
  \draw[->] ( -0.5*\X, 0.5*\Y) -- ( -5+0*\X, 0*\Y);
  \draw[->] ( -0.5*\X, 0*\Y) -- ( -5+0*\X, 0*\Y);
  \draw[->] ( -0.5*\X, -0.5*\Y) -- ( -5+0*\X, 0*\Y);  

  \draw[fill=white] (0*\X, 0*\Y) node{\textbf{A}} +(-5pt, -5pt) rectangle +(5pt, 5pt);
  \draw[fill=white] (1*\X, -1*\Y) node{\textbf{B}} +(-5pt, -5pt) rectangle +(5pt, 5pt);
  \draw (1*\X, 5-1*\Y) node[anchor=south]{$y$};
  \draw[fill=white] (2*\X,  0*\Y) node{\textbf{D}} +(-5pt, -5pt) rectangle +(5pt, 5pt);
  \draw (2*\X, 5-0*\Y) node[anchor=south]{$\pi_2$};

  \draw[->, dashed] (5+0*\X, 0*\Y) -- (-5+2*\X, 0*\Y);

  \draw[->](5+0*\X, 0*\Y) node[anchor=south west]{$[y]_2$} -- 
  (-5+1*\X, -1*\Y); 
  \draw[->](5+1*\X, -1*\Y) --
  (-5+2*\X, 0*\Y); 

  \draw[decorate, decoration={snake, amplitude=0.3mm}](-5+2*\X, 5+0*\Y)
  to[out=180, in=-90] (1.45*\X, 20+0*\Y);
  \draw[->, decorate, decoration={snake, amplitude=0.3mm}](1.45*\X, 20+0*\Y)
  to[out=90+45, in=90-35] node[sloped, above]{$\rho_2$}(5+0*\X, 15+0*\Y);

  \draw[->](5+2*\X, 0*\Y) -- node[above]{$x$}( 2.5*\X, 0*\Y);
\end{tikzpicture} }
    \hspace{20pt}
    \subfloat[Part E][\label{fig:buffersolveE}
    Process~A receives $\rho_2$. It empties the
    buffer with the corresponding counter. Now,  the direct link
    to Process~D is safe. Process~A uses it normally.]
    {
\begin{tikzpicture}[scale=1]
  
  \small
  
  \newcommand\X{210/5pt};
  \newcommand\Y{30pt};
  
  \draw[->] ( -0.5*\X, 0.5*\Y) -- ( -5+0*\X, 0*\Y);
  \draw[->] ( -0.5*\X, 0*\Y) -- ( -5+0*\X, 0*\Y);
  \draw[->] ( -0.5*\X, -0.5*\Y) -- ( -5+0*\X, 0*\Y);  

  \draw[fill=white] (0*\X, 0*\Y) node{\textbf{A}} +(-5pt, -5pt) rectangle +(5pt, 5pt);
  \draw (0*\X, 5-0*\Y) node[anchor=south]{$\rho_2$};
  \draw[fill=white] (1*\X, -1*\Y) node{\textbf{B}} +(-5pt, -5pt) rectangle +(5pt, 5pt);

  \draw[fill=white] (2*\X,  0*\Y) node{\textbf{D}} +(-5pt, -5pt) rectangle +(5pt, 5pt);

  \draw[->] (5+0*\X, 0*\Y) -- node[above]{$y$} (-5+2*\X, 0*\Y);

  \draw[->](5+0*\X, 0*\Y) -- (-5+1*\X, -1*\Y); 
  \draw[->](5+1*\X, -1*\Y) -- node[sloped, above]{$y$} (-5+2*\X, 0*\Y);

  \draw[->](5+2*\X, 0*\Y) -- ( 2.5*\X, 0*\Y);
\end{tikzpicture} }
    \caption{\label{fig:buffersolve}Buffers become bounded. We allow only 2
      elements in each buffer.}
  \end{center}
\end{figure*}

\begin{theorem}[\CBROADCAST is a causal broadcast]
  \CBROADCAST is a causal broadcast in both static and dynamic network settings.
\end{theorem}

\begin{proof}
  For static networks, it comes from~\cite{friedman2004causal}. For dynamic
  networks, it comes from Lemmas~\ref{lem:removals}~and~\ref{lem:additions}.
\end{proof}

\subsection{Bounding space consumption}

Algorithm~\ref{algo:bufferbroadcast} ensures causal delivery of messages even in
dynamic network settings. Compared to the original causal broadcast for static
networks~\cite{friedman2004causal}, it uses an additional local structure:
buffers of messages. It associates a buffer to each new unsafe links. We assumed
that the size of these buffer stays small in general, for it depends on the time
taken by the ping phase which is assumed short. However, network conditions may
invalidate this assumption. Figure~\ref{fig:bufferproblem} depicts scenarios
where buffers grow out of acceptable boundaries. In
Figure~\ref{fig:bufferproblemA}, the issue comes from high transmission delays
from Process~A to Process~B, and from Process~B to Process~D compared to the
number of messages to broadcast and forward. The ping message $\pi$ did not
reach Process~D yet that the buffer contains a lot of messages. In
Figure~\ref{fig:bufferproblemB}, the issue comes from the departure of
Process~D. Depending on network settings, Process~A may not be able to detect
Process~D's departure. The former will never receive the awaited reply and the
buffer will grow forever. In Figure~\ref{fig:bufferproblemC}, the reply $\rho$
itself fails to reach Process~A. For the recall, this message can travel to
Process~A by any communication mean, including unreliable ones. If this fails,
Process~A's buffer to Process~D will grow
forever. 

\begin{algorithm}
  
\SetKwProg{Function}{function}{}{}
\SetKwProg{Upon}{upon}{}{}
\SetKwProg{Initially}{INITIALLY:}{}{}
\SetKwProg{Dissemination}{DISSEMINATION:}{}{}
\SetKwProg{Buffer}{BOUNDING BUFFERS:}{}{}
\SetKwProg{Failure}{HANDLING FAILURES:}{}{}

\small

\DontPrintSemicolon
\LinesNumbered

\Initially {} {
  $B$ \tcp*{link $\rightarrow$ buffered messages}
  \BlankLine
  $I \leftarrow \varnothing$ \tcp*{message id $\leftrightarrow$ link}
  $R \leftarrow \varnothing$ \tcp*{link $\rightarrow$ number of retries}
  \BlankLine
  $maxSize \leftarrow \infty $ \;
  $maxRetry \leftarrow \infty$ \;
}

\BlankLine

\Buffer {} {

  \Upon{$\textup{ping}(from,\, to,\, id)$}{
    \lIf{$q \not\in R$} {$R[q] \leftarrow 0$}
    $I[id] \leftarrow to$ \;
  }

  \BlankLine
  
  \Upon{$\textup{receiveAck}(from,\, to,\, id)$}{
    $I \leftarrow I \setminus id$ \;
    $R \leftarrow R \setminus to$ \;
  }
  
  \BlankLine

  \Upon{$\textup{PC-deliver}(m)$} {
    \ForEach{$q \in B$ \textbf{\textup{such that}} $|B[q]| > maxSize$
      \label{line:maxsize}}{
      $\textup{retry}(q)$ \;
    }
  }

  \BlankLine

  \Upon{$\textup{close}(q)$} {
    \lFor{$i \in I$ \textbf{\textup{such that}} $I[i]=q$}{$I \leftarrow I \setminus i$}
    $R \leftarrow R \setminus q$ \;    
  }

  \BlankLine

  \Function{$\textup{retry}(q)$}{
    \lFor{$i \in I$ \textbf{\textup{such that}} $I[i]=q$}{$I \leftarrow I \setminus i$}
    
    \If{$q \in R$} {
      $R[q] \leftarrow R[q]+ 1$ \;
      \lIf{$R[q] \leq maxRetry$} {\textup{open}($q$)\label{line:reset}}
      \lElse{\textup{close}($q$)}    
    }
  }
}  

\BlankLine

\Failure {} {

  \Upon{$\textup{timeout}(from,\, to,\, id)$}{
    \lIf {$id \in I$} {\textup{retry}($to$)\label{line:timeout}}
  }

}

   \caption{\label{algo:boundingbuffer}Bounding the size of buffers and handling
    network failures.}
\end{algorithm}

Algorithm~\ref{algo:boundingbuffer} solves the unbounded growth issue of
buffers. It solves the issue from the buffer owner's
perspective. Figure~\ref{fig:buffersolve} shows how this algorithm bounds the
size of buffers. In Figure~\ref{fig:buffersolveA}, Process~A broadcast $a$; then
wanted to add a link to Process~D so it sent a ping message; then broadcast $a'$
and $a''$ so it buffered them. We see that the ping message $\pi_1$ carries a
counter. The new buffer is identified by the same counter. In
Figure~\ref{fig:buffersolveB}, Process~A receives, delivers, and forwards the
message $x$. Each message delivery increases the size of current buffers. The
algorithm checks if the size of the buffer exceeds the configured bound (see
Line~\ref{line:maxsize}). Adding $x$ to the buffer would exceed the bound of $2$
elements. This is the first ping phase failure. Process~A simply restarts the
ping phase: it resets the buffer and sends another ping message $\pi_2$ (see
Line~\ref{line:reset}). The counter of the reset buffer is the one of the new
ping message. In the meantime, Process~D receives $\pi_1$ and sends the
corresponding reply $\rho_1$. In Figure~\ref{fig:buffersolveC}, Process~A
receives a broadcast message $y$. It delivers it, checks if the buffer can admit
it, adds the message to the buffer, and forwards it. Process~A also receives the
first reply $\rho_1$ but discards it, for no buffers have such counter. In
Figure~\ref{fig:buffersolveD}, Process~D receives $\pi_2$ and sends the
corresponding reply $\rho_2$ to Process~A. In Figure~\ref{fig:buffersolveE},
Process~A receives $\rho_2$. Since the corresponding buffer exists, it empties
it. The new link is now safe to use for causal broadcast.

While it solves the issue of unbounded buffers, it also brings another
issue. For instance, if the maximal size of buffers is too small, it could stuck
the protocol in a loop of retries. We address this issue by bounding the number
of retries. However, it means that the ping phase could fail
entirely. Causal broadcast must not employ the new link. In extreme cases, it
could cause partitions in the causal broadcast overlay network. It would violate
the uniform agreement property of causal broadcast. Thus, we assume a
sufficiently large maximal bound. It never creates partitions, for most links
become safe, and the failing ones are replaced over time thanks to
network dynamicity.

Other orthogonal improvements are possible. For instance, causal broadcast could
use reliable communication means to acknowledge the receipt of the ping
message. The time taken between the sending and the receipt of the
reply would increase when failures occur. However, it would take less
time than resetting the buffering phase.

\subsection{Complexity}
\label{subsec:complexity}

We review and discuss about the complexity of \CBROADCAST. We distinguish the
complexity brought by 
\begin{inparaenum}[(i)]
\item the overlay network,
\item reliable broadcast,
\item and causal ordering.
\end{inparaenum}

\noindent \textbf{Overlay network.} Processes cannot afford the upkeep of full
membership in large and dynamic systems. Instead, each process builds a partial
view the size of which is considerably smaller than the actual network size.  To
maintain these partial views, each process runs a peer-sampling
protocol~\cite{bertier-d2ht,jelasity2007gossip,jelasity2009tman}.  Some
peer-sampling protocols provides partial views the size of which scales
logarithmically with the actual network size~\cite{nedelec2017adaptive}.  The
number of messages forwarded by each process for each broadcast remains small,
for this number is equal to their view size: $O(Q)$ where $Q$ is the size of the
partial view.

\noindent \textbf{Reliable broadcast.} Gossiping constitutes an efficient mean
to disseminate messages to all
processes~\cite{demers1987epidemic,birman1999bimodal}.
Algorithm~\ref{algo:reliablebroadcast} shows that it relies on a local structure
to guarantee that messages are delivered exactly once. This structure grows
linearly with the number of processes in the network: $O(N)$. In addition, each
message piggybacks a pair $\langle process,\, counter \rangle$ that identifies
it: $O(1)$. Checking if a message is a duplicate takes constant time: $O(1)$.

\noindent \textbf{Causal ordering.} Causal ordering primarily uses FIFO links to
broadcast messages which implies a constant size overhead on messages
$O(1)$. Most space complexity is hidden by FIFO links including that of buffered
messages ensuring safety.
\CBROADCAST maintains one buffer per unsafe link during its ping phase.  We
assume that this time is short so the number of buffered messages stays
small. As shown in Section~\ref{sec:proposal}, network conditions can make this
assumption false. Algorithm~\ref{algo:boundingbuffer} allows to bound the size
of each buffer and handle network failures.

\noindent \textbf{Overall.}  Generated traffic remains the most important
criterion for scalability. The traffic generated by \CBROADCAST for each process
and for each broadcast only depends on the size of messages and the overlay
network chosen to broadcast messages. The size of messages is an irreducible
variable; and many protocols designed to build overlay networks achieve high
scalability in terms of network size and
dynamicity~\cite{bertier-d2ht,jelasity2007gossip,jelasity2009tman,nedelec2017adaptive,voulgaris2005cyclon}.
Consequently, \CBROADCAST achieve high scalability in both these terms
too. \CBROADCAST is efficient, for the upper bound on the complexity of delivery
execution time does not depend on any factor.

However, to ensure causal order, \CBROADCAST may not use all outgoing links in
dynamic settings, for some may be temporarily unsafe. This negatively impacts
the overlay network properties. The next section shows an experiment that
highlights the influence of \CBROADCAST's way to ensure causal order on the
underlying overlay network. In particular, it shows that the number of hops
required by a broadcast message to reach all processes increases when delays on
transmission increase.

\section{Experimentation}
\label{sec:experimentation}

\begin{figure}
  \begin{center}
    \includegraphics[width=0.99\columnwidth]{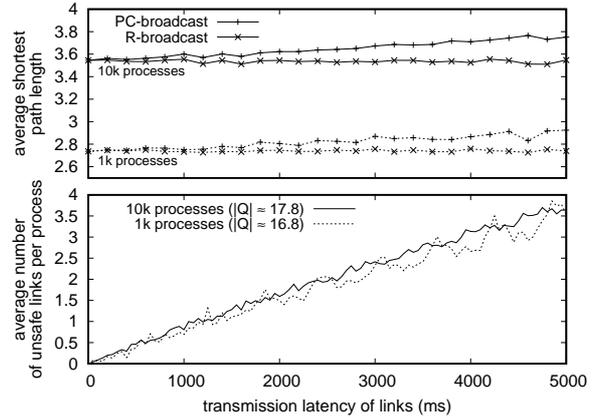}
    \caption{\label{fig:delay}Impact of \CBROADCAST on the overlay network.}
  \end{center}
\end{figure}

\CBROADCAST provides causal order with constant size message
overhead. This feature comes at a cost: at first, new communication
means are disable for causal broadcast. In this section, we evaluate
the impact of \CBROADCAST on the message delivery in a specific
overlay network that corresponds to random graphs. The experiments run
on the \PEERSIM simulator~\cite{montresor2009peersim} that allows
simulations to reach high scale in terms of number of processes. Our
implementation is available on the Github platform at
\url{http://github.com/chat-wane/peersim-pcbroadcast}.

\noindent \textbf{Objective:} To observe the transmission delay introduced by
\CBROADCAST on message delivery. We expect the delay to increase as the latency
increase.

\noindent \textbf{Description:} We build an overlay network with a topology
close to random graphs using \SPRAY~\cite{nedelec2017adaptive}. The overlay
networks comprises 1k, and 10k processes. Networks are dynamic. Each process'
neighborhood $Q$ changes at least once every 60 seconds; and on average twice
every 60 seconds. Each exchange involves two processes that both add and remove
half of their partial view.  Links are FIFO, bidirectional, and have
transmission delays. The delay increases over time up to 5 seconds.
Consequently, the duration of ping phases increases during the experiment.
Links become safe slower. \\
We measure the shortest path length from a random set of processes to all other
processes. It represents the average number of hops taken by broadcast messages
before being received and delivered by all. 
Multiplied by the transmission latency of links, it represents the transmission
delay of broadcast
messages before being received by all processes. \\
We perform measurements on 2 broadcast protocols: \CBROADCAST and R-broadcast.
R-broadcast uses all available links to broadcast messages in a gossip fashion.
Transmission delays before delivery are similar to piggybacking
approaches~\cite{almeida2008interval,fidge1988timestamps,mattern1989virtual,singhal1992efficient,birman1987reliable,hadzilacos1993fault,mostefaoui2017probabilistic}
without accounting for the time taken to send large messages (e.g. each message
convey a vector clock of 10k entries when the network comprises 10k processes).
Larger packets induces larger transmission time.

\noindent \textbf{Results:} Figure~\ref{fig:delay} shows the result of the
experiment. The y-axis depicts the delay set on message transmission for each
link. The top part of the figure shows the average shortest path length. The
bottom part of the figure shows the average number of unsafe links per process
that cannot be used for causal broadcast yet.
\begin{itemize}[leftmargin=*]
\item Figure~\ref{fig:delay} shows that both R-broadcast and \CBROADCAST deliver
  message quickly to all processes. The overlay network guarantee that paths
  stay short and logarithmically scaling with the number of random neighbors in
  partial views.
\item The top part of Figure~\ref{fig:delay} shows that measurements made on
  \CBROADCAST increases while measurements made on R-broadcast stay
  constant. R-broadcast uses all neighbors provided by \SPRAY while \CBROADCAST
  excludes links still in buffering phase. The more latency on transmission, the
  longer the buffering phase. The bottom part of Figure~\ref{fig:delay} shows
  that the number of elements in the buffers increases accordingly.
\item Figure~\ref{fig:delay} shows that the growth of path length stays small
  even when transmission delays become high. The number of elements in buffers
  stays small because the buffering phase takes a constant number of hops to
  complete: at most 3 hops. The path length grows even slower, for removing 3
  among 17 links has restricted impact on overlay networks close to random
  graphs.
\end{itemize}

This experimentation shows that even under bad network conditions (high
transmission delays) and using highly dynamic overlay networks (random
peer-sampling), the number of unsafe links remains small. The negative impact
expected on transmission time before message delivery remains small. In
practice, we expect smaller network transmission delays, and overlay networks
less subject to neighborhood changes (e.g. exploiting user preferences, or
geolocalisation). In such settings, we expect \CBROADCAST to have a negligible
negative impact on the overlay network properties.

The next section reviews state-of-the-art techniques designed to maintain causal
order among messages.

\section{Related work}
\label{sec:relatedwork}

This section reviews the related work of logical clocks. It goes from
piggybacking approaches to vector-based approaches. Then, it reviews explicit
dependency tracking and dissemination-based approaches.

\noindent \textbf{Piggybacking
  approaches~\cite{birman1987reliable,hadzilacos1993fault}}. A trivial way to
ensure causal ordering of messages is to piggyback all causally related messages
since the last broadcast message along with the new broadcast message. Even by
piggybacking the identifiers of messages instead of messages themselves, the
broadcast message size may increase quickly depending on the
application. \CBROADCAST does not piggyback all preceding messages in broadcast
messages. However, an accumulation of messages arises during buffering. As
discussed in Section~\ref{subsec:complexity}, we can assume that links quickly
become safe so the buffer size stays small, and we can easily set a threshold on
the buffer size.

\noindent \textbf{Vector clock
  approaches~\cite{fidge1988timestamps,mattern1989virtual}.}  A vector clock is
a vector of monotonically increasing counters.  It encodes the partial order of
messages using this vector: $VC(m) < VC(m') \implies m \rightarrow m'$.  Before
delivering a message, processes using vector-based broadcast check if the vector
of the message is ready regarding their local vector. If it detects any missing
preceding message, the process delays the delivery.  To implement this
vector-based broadcast
\begin{inparaenum}[(i)]
\item each process must maintain a vector locally;
\item each message must piggyback such vector;
\item there is 1 counter per process that ever broadcast a message.
\end{inparaenum}
To accurately track causality, processes cannot share their entry. To safely
track causality, processes cannot reclaim entries. Hence, even with
\textbf{compaction approaches~\cite{singhal1992efficient}}, the vectors grow
linearly in terms of number of processes that ever broadcast a message.
In~\cite{almeida2008interval}, the complexity is reduced to the actual number of
processes in the network.  Still, these approaches do not scale, particularly in
dynamic networks subject to churn and failures. \\
In comparison to these vector-based approaches, our approach reduces the
generated traffic of causal broadcast by a factor of $N$ in the most common
context where processes have partial knowledge of the network membership. \\
\textbf{Probabilistic approaches~\cite{mostefaoui2017probabilistic}} sacrifices
on causality tracking accuracy: messages may be delivered out of order under a
computable boundary. The size of control information in messages
depends on the desired boundary. \\
Unlike vector-based approach, our broadcast cannot state if two messages are
concurrent, accurate causal delivery is a feature provided by default by the
propagation scheme. Once safe, FIFO links deliver message in causal order
without further delay. The speed of delivery is that of FIFO links.

\noindent Explicitly tracking \textbf{semantic dependencies} allows broadcast
protocols to reduce the size of piggybacked control
information~\cite{bailis2013bolton,lloyd2011cops,mukund2014optimized}. For
instance, when Alice comments Bob's picture, everyone must receive the picture
before the comment. The broadcast message only conveys one semantic
dependency. When Alice comments multiple pictures at once, the broadcast message
conveys all dependencies.  Message overhead increases linearly with the number
of semantic dependencies. To track semantic dependencies, causal broadcast
becomes application dependent. Instead our approach remains
application-agnostic. Comments, pictures, etc. are events that relate to all
preceding events. When Alice comments Bob's picture, everyone will receive this
event before the former event and all other preceding events. Whatever the
number of preceding events, broadcast messages only convey constant size control
information.

\noindent Preserving causal order using \textbf{dissemination paths} reduces
generated traffic by keeping message overhead
constant~\cite{bravo2017saturn,friedman2004causal}. State-of-the-art does not
support dynamic systems~\cite{friedman2004causal}, or supports it using
epochs~\cite{bravo2017saturn} that confines usability to small scale systems
where failures are uncommon. In comparison, we designed \CBROADCAST to handle
large and dynamic systems. Our approach provides a lightweight and efficient
mean to reconfigure dissemination paths using local knowledge without impairing
causal order.  Saturn~\cite{bravo2017saturn} along with \CBROADCAST could ease
online changes in configuration while improving its resilience to failures and topology
changes.

\section{Conclusion}
\label{sec:conclusion}

In this paper, we described a non-blocking causal broadcast protocol that breaks
scalability barriers for large and dynamic systems. Using \CBROADCAST, message
overhead and delivery execution time remain constant.
Causal broadcast finally becomes an affordable and efficient middleware for
large scale distributed applications running in dynamic environments.

As future work, we plan to investigate on reducing the space complexity of
reliable broadcast. Section~\ref{subsec:complexity} reviews structures with
linearly increasing space consumption. We can reduce this complexity in static
systems. We can prune the structure from already received messages, for we know
that the number of duplicates is equal to the number of incoming
links~\cite{raynal2013distributed}. Unfortunately, this does not hold in dynamic
systems. We would like to investigate on a way to prune the structure in such
settings. This would make causal broadcast scalable as well on generated traffic
as on space consumption.

We also plan to investigate on retrieving partial order of
events. Section~\ref{sec:relatedwork} states that vector-based approaches allows
to compare an event with any other event. They can decide on whether one
precedes the other, or they are concurrent. They can build the partial order of
event using this knowledge. Our approach cannot by default. However, in extreme
settings where the overlay network is fully connected, we can assign a vector to
each received message using local knowledge only, and without message
overhead. We would like to investigate on a way to build these vectors locally
in more realistic settings where processes have partial knowledge of the
membership.

\bibliographystyle{IEEEtran}
\bibliography{bibliographie}
  
\end{document}